%% file: sublinear_time.tex
\documentclass[12pt, draftclsnofoot, onecolumn]{IEEEtran}

\usepackage{color}
\usepackage{cite}
\usepackage[cmex10]{amsmath}
\usepackage{amssymb}
\usepackage{array}
\usepackage{wasysym}
\usepackage{dsfont}
\usepackage[latin1]{inputenc}                       
\usepackage[english]{babel}                         
\usepackage[T1]{fontenc}
\usepackage{mathtools}
\usepackage{amssymb} 
\usepackage{enumerate}
\usepackage{bbm}
\usepackage{epsfig,syntonly}
\usepackage{verbatim,times}
\usepackage{epstopdf}
\usepackage{graphicx}
\usepackage{latexsym,fancyhdr,bm}
\usepackage{subcaption}
\usepackage{url}
\usepackage{bbm}
\usepackage{dsfont}
\usepackage{amsthm}

\usepackage{tikz}
\usetikzlibrary{arrows,positioning,shapes.geometric,calc}
\usepackage{pgfplots}
\pgfplotsset{compat=1.14}

\newtheoremstyle{custom}
{} 
{} 
{} 
{} 
{\bfseries} 
{:} 
{.25em} 
{} 
\theoremstyle{custom}

\newtheorem{theorem}{Theorem}
\newtheorem{lemma}{Lemma}

\newtheorem{definition}{Definition}

\newtheorem{remark}{Remark}

\newtheorem*{theorem*}{Theorem}
\newtheorem*{lemma*}{Lemma}
\newtheorem*{proposition*}{Proposition}
\newtheorem*{definition*}{Definition}
\newtheorem*{example*}{Example}
\newtheorem*{remark*}{Remark}
\newtheorem*{corollary*}{Corollary}

\makeatletter
\let\l@ENGLISH\l@english
\makeatother

\title{Sublinear Latency for Simplified Successive Cancellation Decoding of Polar Codes}

\author{Marco~Mondelli, Seyyed~Ali~Hashemi, John~Cioffi, Andrea~Goldsmith
\thanks{M.~Mondelli is with the Institute of Science and Technology (IST) Austria, Klosterneuburg, Austria (email: marco.mondelli@ist.ac.at). S.~A.~Hashemi and J.~Cioffi are with the Department of Electrical Engineering, Stanford University, Stanford, CA 94305, USA (email: ahashemi@stanford.edu, cioffi@stanford.edu). A.~Goldsmith is with the Department of Electrical Engineering, Princeton University, Princeton, NJ 08544, USA (email: goldsmith@princeton.edu).}
}

\begin{document}

\maketitle
\begin{abstract}
\noindent This work analyzes the latency of the simplified successive cancellation (SSC) decoding scheme for polar codes proposed by Alamdar-Yazdi and Kschischang. It is shown that, unlike conventional successive cancellation decoding, where latency is linear in the block length, the latency of SSC decoding is sublinear. More specifically, the latency of SSC decoding is $O(N^{1-1/\mu})$, where $N$ is the block length and $\mu$ is the scaling exponent of the channel, which captures the speed of convergence of the rate to capacity. Numerical results demonstrate the tightness of the bound and show that most of the latency reduction arises from the parallel decoding of subcodes of rate $0$ or $1$.
\end{abstract}


\section{Introduction} \label{sec:intro}

Polar codes provably achieve capacity for any binary memoryless symmetric (BMS) channel with low encoding and decoding complexity \cite{Ari09}. Because of their attractive properties, polar codes have been recently adopted for the enhanced mobile broadband (eMBB) control channel of the fifth generation (5G) wireless communications standard \cite{3gpp_polar}. For a polar code of block length $N$, the encoding and decoding complexity is $O(N\log N)$; the code construction can be performed with complexity $O(N)$ \cite{TV13con, RHTT} and, by exploiting a partial order between the synthetic channels, the construction complexity becomes sublinear in $N$ \cite{mondelli2018construction}. In addition, the error probability under successive cancellation (SC) decoding scales with the block length roughly as $2^{-\sqrt{N}}$ \cite{ArT09}. Moreover, polar codes are not affected by error floors \cite{MHU15unif-ieeeit}.

The speed of convergence of the rate to capacity has also been extensively studied \cite{HAU14, MHU15unif-ieeeit, XG13, GB14, MHU14list-ieeeit, fazeli2018binary, guruswami2019ar}. These works demonstrate that the gap to capacity scales with the block length as $N^{-1/\mu}$, where the parameter $\mu$ is called the \emph{scaling exponent} and it depends on the transmission channel. Equivalently, the smallest block length needed to achieve an assigned gap to capacity scales as
\begin{equation}
N \sim \frac{1}{(I(W)-R)^\mu} \text{,}
\end{equation}
where $R$ is the rate of the code and $I(W)$ is the capacity of the BMS channel $W$. For any BMS $W$, the following upper and lower bounds on $\mu$ hold: $3.579 \le \mu\le 4.714$. Furthermore, when $W$ is a binary erasure channel (BEC), then $\mu \approx 3.63$ \cite{HAU14, MHU15unif-ieeeit}; when $W$ is a binary additive white Gaussian noise channel (BAWGNC), then $\mu \approx 4$ \cite{KMTU10}; and when $W$ is a binary symmetric channel (BSC), a conjecture is that $\mu \approx 4.2$. The introduction of any finite-size list does not improve the scaling exponent under optimal MAP decoding and genie-aided SC decoding \cite{MHU14list-ieeeit}. However, by using large polarization kernels, it is possible to approach the optimal scaling exponent $\mu=2$ \cite{fazeli2018binary, guruswami2019ar}. The moderate deviations regime, in which both the error probability and the gap to capacity jointly vanish as the block length grows large, has also been a subject of recent investigation \cite{MHU15unif-ieeeit, fong2017scaling, wang2018polar, blasiok2018polar}.

In \cite{TVa15} the error correction performance of the SC decoder is improved through an SC list (SCL) decoder with time complexity $O(L N\log N)$ and space complexity $O(L N)$, where $L$ is the size of the list. SCL decoding keeps a list of the most likely codewords by running $L$ coupled SC decoders in parallel. Empirically, the error probability of the SCL decoder is close to that of the optimal MAP decoder with practical values of the list size $L$. Furthermore, by adding a few extra bits of cyclic redundancy check (CRC) precoding, the performance is comparable to state-of-the-art low-density parity-check (LDPC) codes. One disadvantage of SCL decoding is the large area required in hardware since multiple coupled SC decoders need to be implemented. Partitioned SCL decoders have been proposed to address this issue \cite{hashemi_PSCL, hashemi2017partitioned, hashemi2018decoder}.

Another problem associated with SC-based decoding algorithms, such as SC and SCL, is their high latency. In fact, SC decoding is a serial algorithm, in the sense that decoding proceeds bit by bit. In order to address the problem, a \emph{simplified} SC (SSC) decoder was proposed in \cite{alamdar}, which identifies smaller constituent codes in the polar code and decodes them in parallel. As a result, the latency is reduced with no penalty in the error correction performance. In \cite{sarkis,hanif}, more constituent codes were identified and low-complexity parallel decoders were designed to increase the throughput and reduce the latency even further. In \cite{hashemi_SSCL_TCASI,hashemi_FSSCL_TSP}, these results were extended to SCL decoding. This extension introduced a simplified SCL (SSCL) algorithm that decodes the constituent codes in parallel while keeping the same error correction performance as the standard SCL decoding. Recently, a variant of polar codes with log-logarithmic time complexity per information bit has been introduced in \cite{wang2019log}. This improves upon the logarithmic time complexity per information bit for SC decoding of standard polar codes. However, the time complexity per information bit is a different metric from the decoding latency, which represents the time complexity of the overall decoding process.

This paper quantifies the latency of the SSC decoder proposed in \cite{alamdar}. The main result is that the number of time steps needed by the SSC decoder is $O(N^{1-1/\mu})$, which results in sublinear latency. As a benchmark, the decoding latency of the standard SC decoder with a fully parallel architecture is $2N-1$ time steps \cite{Ari09}, thus the SSC decoder yields a (multiplicative) latency gain of $N^{1/\mu}$, where $\mu$ is the aforementioned scaling exponent. To be concrete, this means that the latency of SSC decoding of polar codes scales roughly as $N^{3/4}$ (more precisely, it scales as $N^{0.72}$ for BECs and as $N^{0.76}$ for BSCs). Numerical results show that this bound is tight and also captures the dependence on the transmission channel via the scaling exponent.

The remainder of the paper is organized as follows: Section~\ref{sec:prel} provides some preliminaries that include the formal definition of scaling exponent, the construction rule, and the SC and SSC decoding algorithms; Section~\ref{sec:main} states and proves that the latency of SSC decoding is $O(N^{1-1/\mu})$, deferring the proofs of two intermediate lemmas to Appendix~\ref{app:proofs}; Section~\ref{sec:numerical} presents some numerical results that demonstrate the tightness of the upper bound; and Section \ref{sec:concl} concludes the paper. The numerical results also show that most of the savings arises from pruning constituent codes that are either rate-$0$ or rate-$1$: pruning additional constituent codes provides some latency gain at moderate block lengths, but it is suggested that the latency still scales as $N^{1-1/\mu}$ for large $N$. 

\section{Polar Coding Preliminaries} \label{sec:prel}

\subsection{Channel Polarization}

Let $W$ be a BMS channel with input alphabet $\mathcal{X}=\{0,1\}$, output alphabet 
$\mathcal{Y}$, and transition probabilities $\{W(y \mid x) : x\in \mathcal{X}, y\in \mathcal{Y}\}$. Denote by $Z(W)\in [0,1]$ the Bhattacharyya parameter of $W$, which
is defined as
\begin{align*}
& Z(W)= \sum_{y \in \mathcal{Y}} \sqrt{W(y\mid 0)W(y \mid 1)}.
\end{align*}
$Z(W)$ is a measure of the reliability of $W$: if $Z(W)\approx 0$, then the channel is almost noiseless (i.e., its capacity $I(W)\approx 1$); and if $Z(W)\approx 1$, then the channel is very noisy (i.e., its capacity $I(W)\approx 0$). The basis of channel polarization is to map two identical copies of the channel $W: \mathcal{X}\to \mathcal{Y}$ into the pair of channels $W^0: \mathcal{X}\to \mathcal{Y}^2$ and $W^1:\mathcal{X}\to \mathcal{X}\times\mathcal{Y}^2$, defined as \cite[Section I-B]{Ari09}, \cite[Section I-B]{HAU14},
\begin{equation}\label{eq:mapch}
\begin{split}
W^0(y_1, y_2\mid x_1) & = \sum_{x_2\in \mathcal X} \frac{1}{2}W(y_1\mid x_1 \oplus x_2) W(y_2\mid x_2),\\
W^1(y_1, y_2, x_1\mid x_2) & = \frac{1}{2}W(y_1\mid x_1 \oplus x_2) W(y_2\mid x_2).\\
\end{split}
\end{equation}
Then, the idea is that $W^0$ is a ``worse'' channel and $W^1$ is a ``better'' channel than $W$. This statement can be quantified by the following bounds among the Bhattacharyya parameters of $W$, $W^0$, and $W^1$:
\begin{align}
Z(W)\sqrt{2-Z(W)^2}&\le Z(W^0)\le 2Z(W)-Z(W)^2,\label{eq:minusB}\\
&Z(W^1)=Z(W)^2,\label{eq:plusB}
\end{align}
which follow from Proposition 5 of \cite{Ari09} and from Exercise 4.62 of \cite{RiU08}. By repeating $n$ times the operation \eqref{eq:mapch}, we map $2^n$ identical copies of $W$ into the synthetic channels $W_n^{(i)}$ ($i\in \{1, \ldots, 2^n\}$), defined as 
\begin{equation}\label{eq:defWni}
W_n^{(i)} = (((W^{b_1^{(i)}})^{b_2^{(i)}})^{\cdots})^{b_n^{(i)}},
\end{equation}
where $(b_1^{(i)}, \ldots, b_n^{(i)})$ is the binary representation of the integer $i-1$ over $n$ bits. Furthermore, define a random sequence of channels $W_n$, as $W_0=W$, and 
\begin{equation}
W_{n} = \left\{ \begin{array}{ll}W_{n-1}^0, & \mbox{ w.p. } 1/2,\\ W_{n-1}^1,& \mbox{ w.p. } 1/2.\\ \end{array}\right.
\end{equation}
Let $Z_n(W)=Z(W_n)$ be the random process that tracks the Bhattacharyya parameter of $W_n$. Then, from \eqref{eq:minusB} and \eqref{eq:plusB} we deduce that, for $n\ge 1$,
\begin{equation}\label{eq:eqBMSC}
Z_{n} \left\{ \begin{array}{ll}\in \left[Z_{n-1}\sqrt{2-Z^2_{n-1}},\, 2 Z_{n-1}-Z^2_{n-1}\right], & \mbox{ w.p. } 1/2,\\ =Z^2_{n-1},& \mbox{ w.p. } 1/2.\\ \end{array}\right.
\end{equation}

The synthetic channels $W_n^{(i)}$ \emph{polarize} in the sense that, as $n$ grows large, most of them become either completely noisy or completely noiseless. Then, we put information bits in the noiseless synthetic channels, and we \emph{freeze} to 0 the remaining ones. Formally, as $n\to\infty$, $Z_n$ converges almost surely to a random variable $Z_\infty$ such that
\begin{equation}
    Z_\infty=\left\{ \begin{array}{ll}0, & \mbox{ w.p. } I(W),\\ 1,& \mbox{ w.p. } 1-I(W).\\ \end{array}\right.
\end{equation}

\subsection{Scaling Exponent}

The fact that the synthetic channels $W_n^{(i)}$ are ``polarized'' implies that polar codes achieve capacity. The scaling exponent captures the speed of convergence as $N$ increases. 

\begin{definition}[Upper bound on scaling exponent]\label{def:upscal}
We say that $\mu$ is an \emph{upper bound on the scaling exponent} if there exists a function $h(x): [0, 1] \to [0, 1]$ such that $h(0)=h(1)=0$, $h(x)>0$ for any $x\in (0, 1)$, and
\begin{equation}\label{eq:suph}
\displaystyle\sup_{\substack{x\in (0, 1), y \in [x\sqrt{2-x^2}, 2x-x^2]}}\displaystyle\frac{h(x^2)+h(y)}{2h(x)} < 2^{-1/\mu}.
\end{equation}
\end{definition}

\begin{definition}[Upper bound on scaling exponent of BEC]\label{def:ubscalBEC}
We say that $\mu$ is an \emph{upper bound on the scaling exponent of BEC} if there exists a function $h(x): [0, 1] \to [0, 1]$ such that $h(0)=h(1)=0$, $h(x)>0$ for any $x\in (0, 1)$, and
\begin{equation}\label{eq:suphBEC}
\displaystyle\sup_{x\in (0, 1)}\displaystyle\frac{h(x^2)+h(2x-x^2)}{2h(x)} < 2^{-1/\mu}.
\end{equation}
\end{definition}

The definitions above are motivated by \cite[Theorem 1]{MHU15unif-ieeeit}, where it is shown that if $\mu$ is an upper bound on the scaling exponent according to Definition \ref{def:upscal}, then the gap to capacity $I(W)-R$ scales with the block length as $N^{-1/\mu}$. Note that, when the transmission channel is a BEC, then $Z(W^0)= 2Z(W)-Z(W)^2$. Consequently, the condition \eqref{eq:suph} is replaced by \eqref{eq:suphBEC}, see Definition~\ref{def:ubscalBEC}. Valid choices of upper bounds on the scaling exponent are $\mu=4.714$ and $\mu=3.639$ for the special case of BEC, as shown in \cite[Theorem 2]{MHU15unif-ieeeit}.

\subsection{Construction}

\begin{definition}[Polar code construction]\label{def:construction}
Let $p_{\rm e}\in (0, 1)$, $W$ be a BMS channel, and $N=2^n$ be the polar code block length. Then the polar code $\mathcal C_{\rm polar}(p_{\rm e}, W, N)$ is obtained by placing the information bits into the positions corresponding to all the synthetic channels whose Bhattacharyya parameter is less than $p_{\rm e}/N$ and by freezing the remaining positions.
\end{definition}

The construction rule of Definition \ref{def:construction} ensures that the error probability under SC decoding  is \emph{at most} $p_{\rm e}$. In fact, the error probability can be upper bounded by the sum of the Bhattacharyya parameters of the synthetic channels associated with the information bits (cf. Proposition 2 of \cite{Ari09}), and each of them is at most $p_{\rm e}/N$. Furthermore, this construction rule also ensures that the rate $R$ of the code tends to capacity at a speed captured by the scaling exponent. In particular, by using \cite[Theorem 1]{MHU15unif-ieeeit}, the gap to capacity $I(W)-R$ is $O(N^{-1/\mu})$, where $\mu$ is an upper bound on the scaling exponent according to Definition \ref{def:upscal} (for the special case of BEC, see Definition \ref{def:ubscalBEC}).

\subsection{Successive Cancellation Decoding}

SC decoding can be equated to passing messages on a binary tree, as shown in Figure~\ref{fig:scDec}, with priority given to the left branches. Two kinds of messages are passed between the nodes at the binary tree: the logarithmic likelihood ratio (LLR) values that are passed from the top to the bottom of the tree, and the hard bit estimations that are passed from the bottom to the top of the tree. At each node at level $s$ of the SC decoding tree, the LLR values $\bm{\alpha} = \{\alpha_1,\ldots,\alpha_{2^{s+1}}\}$ are received from a parent node at level $s+1$. The LLR values $\bm{\alpha}$ are used to calculate the LLR values of the left child node $\bm{\alpha}^{\ell} = \{\alpha^{\ell}_1,\ldots,\alpha^{\ell}_{2^{s}}\}$ and the right child node $\bm{\alpha}^{\text{r}} = \{\alpha^{\text{r}}_1,\ldots,\alpha^{\text{r}}_{2^{s}}\}$. Furthermore, the hard bit estimations $\bm{\beta} = \{\beta_1,\ldots,\beta_{2^{s+1}}\}$ are calculated based on the hard bit estimations that are received from the left child node $\bm{\beta}^{\ell} = \{\beta^{\ell}_1,\ldots,\beta^{\ell}_{2^{s}}\}$ and the right child node $\bm{\beta}^{\text{r}} = \{\beta^{\text{r}}_1,\ldots,\beta^{\text{r}}_{2^{s}}\}$ in accordance with
\begin{align}
    \alpha^{\ell}_i &= f^{\ell}_s(\alpha_i,\alpha_{i+2^s}) \text{,} \label{eq:flFunc}\\
    \alpha^{\text{r}}_i &= f^{\text{r}}_s(\alpha_i,\alpha_{i+2^s},\beta^{\ell}_i) \text{,} \label{eq:frFunc}\\
    \beta_i &=
    \begin{cases}
    \beta^{\ell}_i \oplus \beta^{\text{r}}_i, & \text{if } i \leq 2^s \text{,} \\
    \beta^{\text{r}}_{i-2^s}, & \text{otherwise,}
    \end{cases}
\end{align}
where $\oplus$ is the XOR operator and the functions $f^{\ell}_s$ and $f^{\text{r}}_s$ are defined as
\begin{align}
    f^{\ell}_s(a,b) &= \ln{\frac{1+e^{a+b}}{e^a+e^b}} \text{,} \\
    f^{\text{r}}_s(a,b,c) &= b+(1-2c)a \text{.}
\end{align}
At a leaf node of the SC decoding tree, each bit $\hat{u}_i$ is estimated as
\begin{equation}
    \hat{u}_i =
    \begin{cases}
    0, & \text{if $u_i$ is a frozen bit or $\alpha^0_i>0$,} \\
    1, & \text{otherwise,}
    \end{cases}
\end{equation}
where $\alpha^0_i$ is the calculated LLR value of $u_i$. The value of $\hat{u}_i$ is used to update the hard bit estimations at the higher levels of the decoding tree.

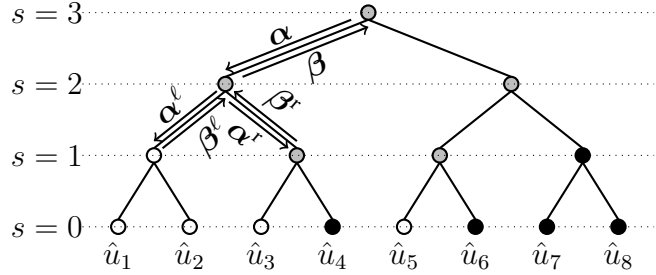
\begin{figure}[t]
\centering
\input{figures/sc-dec}
\caption{SC decoding tree for a polar code with $N=8$ and $R=1/2$. The white nodes represent Rate-0 nodes, the black nodes represent Rate-1 nodes, and the gray nodes are neither Rate-0 nodes nor Rate-1 nodes.}
\label{fig:scDec}
\end{figure}

SC has a sequential structure in the sense that the decoding of each bit is dependent on the decoding of its previous bits. More formally, while the function $f^{\ell}_s$ at level $s$ is only dependent on the LLR values that are received from a parent node ($\alpha_i$ and $\alpha_{i+2^s}$), the function $f^{\text{r}}_s$ is also dependent on a hard bit estimation ($\beta^{\ell}_i$) that is a result of estimating the previous bits (see (\ref{eq:flFunc}) and (\ref{eq:frFunc})). As a result, SC decoding proceeds by traversing the binary tree such that the nodes at level $s=0$ are visited from left to right. For example, in the SC decoding tree of Figure~\ref{fig:scDec} for a polar code of length $N=8$, the following schedule in performing $f^{\ell}_s$ and $f^{\text{r}}_s$ will complete the decoding process:
\begin{equation}
    \text{channel} \rightarrow f^{\ell}_2 \rightarrow f^{\ell}_1 \rightarrow f^{\ell}_0 \rightarrow f^{\text{r}}_0 \rightarrow f^{\text{r}}_1 \rightarrow f^{\ell}_0 \rightarrow f^{\text{r}}_0 \rightarrow f^{\text{r}}_2 \rightarrow f^{\ell}_1 \rightarrow f^{\ell}_0 \rightarrow f^{\text{r}}_0 \rightarrow f^{\text{r}}_1 \rightarrow f^{\ell}_0 \rightarrow f^{\text{r}}_0 \text{,} \label{eq:sch}
\end{equation}
where ``channel'' refers to the time step needed to retrieve channel LLR values. 

Note that the operations at each node of the tree can be performed in parallel. Thus, in a fully parallel SC decoder architecture \cite{Ari09}, the scheduling in (\ref{eq:sch}) for a polar code of length $N$ results in $2N-1$ time steps. This corresponds to the number of nodes in the SC decoding tree.

\subsection{Simplified Successive Cancellation Decoding}\label{subsec:simplified}

The sequential decoding nature of SC decoding results in high latency and low throughput when used to decode polar codes. An SSC decoding algorithm was proposed in \cite{alamdar} by identifying two types of nodes in the SC decoding tree that can be decoded efficiently without traversing their child nodes. These two node types are defined as follows:
\begin{itemize}
    \item \emph{Rate-0 node}: A Rate-0 node at level $s$ of the SC decoding tree is such that all its leaf nodes at level $0$ are frozen bits. Since the values of frozen bits are known to the decoder, there is no need to traverse the decoding tree below Rate-0 nodes and the hard bit estimations can be directly calculated at level $s$ where the Rate-0 node is located. For a Rate-0 node at level $s$ we have
    \begin{equation}
        \beta^s_i = 0 \text{,}
    \end{equation}
    where $\beta^s_i$ is the hard bit estimation of the $i$-th bit.
    \item \emph{Rate-1 node}: A Rate-1 node at level $s$ of the SC decoding tree is such that all its leaf nodes at level $0$ are information bits. It was shown in \cite{alamdar} that there is no need to traverse the decoding tree below Rate-1 nodes and the hard bit estimations can be directly calculated at level $s$ where the Rate-1 node is located. For a Rate-1 node at level $s$ we have
    \begin{equation}
        \beta^s_i =
        \begin{cases}
        0 & \text{if $\alpha^s_i>0$,}\\
        1 & \text{otherwise,}
        \end{cases}
    \end{equation}
    where $\beta^s_i$ is the hard bit estimation and $\alpha^s_i$ is the LLR value.
\end{itemize}
In fact, SSC decoding can decode Rate-0 and Rate-1 nodes in a single time step. In a binary tree representation of SC decoding, this corresponds to pruning all the nodes that are the descendants of a Rate-0 node or a Rate-1 node. This is illustrated in Figure~\ref{fig:sscDec} for the same example as in Figure~\ref{fig:scDec}. The SSC decoding schedule for decoding the example in Figure~\ref{fig:sscDec} is:
\begin{equation}
    \text{channel} \rightarrow f^{\ell}_2 \rightarrow f^{\ell}_1 \rightarrow f^{\text{r}}_1 \rightarrow f^{\ell}_0 \rightarrow f^{\text{r}}_0 \rightarrow f^{\text{r}}_2 \rightarrow f^{\ell}_1 \rightarrow f^{\ell}_0 \rightarrow f^{\text{r}}_0 \rightarrow f^{\text{r}}_1 \text{,} \label{eq:schSSC}
\end{equation}
which requires four fewer time steps in comparison with the required number of time steps for SC decoding in (\ref{eq:sch}). For practical code lengths, SSC has significantly lower latency than SC decoding \cite{alamdar}. This is due to the fact that the number of nodes in the SSC decoding tree is significantly less than the number of nodes in the SC decoding tree.

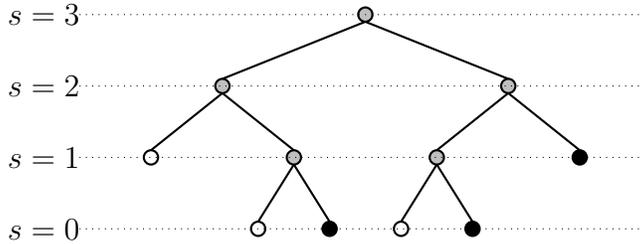
\begin{figure}[t]
\centering
\input{figures/ssc-dec}
\caption{SSC decoding tree for a polar code with $N=8$ and $R=1/2$. Note that Rate-0 and Rate-1 nodes in the SC decoding tree are pruned to get the SSC decoding tree.}
\label{fig:sscDec}
\end{figure}

\section{Upper Bound on the Latency of SSC Decoding} \label{sec:main}

\begin{theorem}[Sublinear latency with SSC decoding]\label{th:main}
Let $p_{\rm e}\in (0, 1)$, $W$ be a BMS channel, and $N=2^n$ be the polar code block length. Consider the polar code $\mathcal C_{\rm polar}(p_{\rm e}, W, N)$ constructed according to Definition \ref{def:construction}. Let $\mu$ be an upper bound on the scaling exponent according to Definition \ref{def:upscal}. Then, the latency of the SSC decoder is $O(N^{1-1/\mu})$.
\end{theorem}

\begin{remark}[Sublinear latency with SSC decoding for BEC]
For the special case of BEC, the latency of the SSC decoder is $O(N^{1-1/\mu})$, where $\mu$ is an upper bound on the scaling exponent of BEC according to Definition \ref{def:ubscalBEC}.
\end{remark}

The proof relies on two intermediate results, which are stated below and proved in Appendix~\ref{app:proofs}. The first intermediate result provides an accurate bound on the fraction of synthetic channels that are un-polarized in the sense that their Bhattacharyya parameters are not too small and not too large. A tighter result holds for BEC, see Remark \ref{rmk:BEC} in Appendix~\ref{app:proofs}.

\begin{lemma}[Number of un-polarized channels]\label{lemma:unpolarized}
Let $W$ be a BMS channel and let 
$Z_n=Z(W_n)$ be the random process that tracks the Bhattacharyya parameter of $W_n$. Let $\mu$ be an upper bound on the scaling exponent according to Definition \ref{def:upscal}. Fix a constant $\nu>1$. Then, for $n\ge 1$, 
\begin{equation}\label{eq:nummid}
    \mathbb P(Z_n\in [2^{-\nu n}, 1-2^{- \nu n}]) \le c\,2^{-n/\mu},
\end{equation}
where the constant $c$ depends solely on $\nu$ and it does not depend on $n$ or $W$.
\end{lemma}

If the transmission channel is almost noiseless or very noisy, then the rate of the corresponding polar code is $1$ or $0$, respectively. The second intermediate result quantifies this statement by providing sufficient conditions on the Bhattacharyya parameter of the channel so that the polar code has either rate $0$ or rate $1$.

\begin{lemma}[Sufficient condition for Rate-0 and Rate-1 nodes]\label{lemma:rate01}
Let $W$ be a BMS channel, $p_{\rm e}\in (0, 1)$, $N=2^n$, and $M=2^m$ with $m<n$. Consider the polar code $\mathcal C_{\rm polar}(p_{\rm e}/M, W, N/M)$ constructed according to Definition \ref{def:construction}. Then, there exists an integer $n_0$, which depends on $p_{\rm e}$, such that for $n\ge n_0$, the following holds:
\begin{enumerate}
    \item If $Z(W)\le 1/N^3$, then the polar code $\mathcal C_{\rm polar}(p_{\rm e}/M, W, N/M)$ has rate $1$.
    \item If $Z(W)\ge 1-1/N^3$, then the polar code $\mathcal C_{\rm polar}(p_{\rm e}/M, W, N/M)$ has rate $0$.
\end{enumerate} 
\end{lemma}

At this point, we are ready to prove Theorem~\ref{th:main}.

\begin{proof}[Proof of Theorem~\ref{th:main}]
From the discussion in Section~\ref{subsec:simplified}, it suffices to show that the number of nodes of the SSC decoding tree is $O(N^{1-1/\mu})$. As the block length of the code is $N=2^n$, the synthetic channels go through $n$ steps of polarization, or equivalently, the depth of the decoding tree is $n$. These $n$ polarization steps are divided into $K$ rounds. For $k\in \{1, \ldots, K\}$, the $k$-th round contains $\delta_k n$ polarization steps, with $\sum_{k=1}^K\delta_k=1$. The idea is that, at the end of each round, the number of un-polarized synthetic channels is given by Lemma~\ref{lemma:unpolarized}. The remaining synthetic channels are polarized in the sense that their Bhattacharyya parameter is very close to $0$ or to $1$. Thus, by Lemma~\ref{lemma:rate01}, these synthetic channels lead to Rate-0 or Rate-1 nodes, which can be pruned. 

More formally, after $\delta_1 n$ steps of polarization, there are a total of $N^{\delta_1}$ synthetic channels. By applying Lemma~\ref{lemma:unpolarized} with $\nu=3/\delta_1$, we have that at most $c(\delta_1) N^{\delta_1(1-1/\mu)}$ of these channels have a Bhattacharyya parameter that belongs to the interval $[N^{-3}, 1-N^{-3}]$. $c(\delta_1)$ is a constant that depends uniquely on $\delta_1$ (and not on $N$ or $W$). The Bhattacharyya parameter of the remaining synthetic channels is either smaller than $1/N^3$ or larger than $1-1/N^3$. Thus, by applying Lemma~\ref{lemma:rate01} with $M=N^{\delta_1}$, these remaining synthetic channels are Rate-0 or Rate-1 nodes, and they can be pruned. After pruning, the remaining number of nodes is
\begin{equation}\label{eq:remaining1}
    O(N^{\delta_1}+N^{\delta_1(1-1/\mu)+1-\delta_1}).
\end{equation}
In fact, the term $O(N^{\delta_1})$ in \eqref{eq:remaining1} comes from the fact that no pruning takes places in the first $\delta_1 n-1$ steps of polarization; and the term $O(N^{\delta_1(1-1/\mu)+1-\delta_1})$ comes from the fact that, after pruning, there are $O(N^{\delta_1(1-1/\mu)})$ remaining nodes at depth $\delta_1 n$, and each of these nodes is the root of a tree containing $2N^{1-\delta_1}-1$ nodes.   

The same procedure is repeated with each of the $O(N^{\delta_1(1-1/\mu)})$ remaining nodes at depth $\delta_1 n$. Consider one of these nodes. After $\delta_2 n$ steps of polarization, there are a total of $N^{\delta_2}$ synthetic channels. By applying Lemma~\ref{lemma:unpolarized} with $\nu=3/\delta_2$, we have that at most $c(\delta_2) N^{\delta_2(1-1/\mu)}$ of these channels have a Bhattacharyya parameter that belongs to the interval $[N^{-3}, 1-N^{-3}]$. The Bhattacharyya parameter of the remaining synthetic channels is either smaller than $1/N^3$ or larger than $1-1/N^3$. Thus, by applying Lemma~\ref{lemma:rate01} with $M=N^{\delta_1+\delta_2}$, these remaining synthetic channels are Rate-0 or Rate-1 nodes, and they can be pruned. The remaining number of nodes is given by
\begin{equation}\label{eq:remaining2}
    O(N^{\delta_1}+N^{\delta_1(1-1/\mu)+\delta_2}+N^{(\delta_1+\delta_2)(1-1/\mu)+1-\delta_1-\delta_2}).
\end{equation}
In fact, the term $O(N^{\delta_1})$ in \eqref{eq:remaining2} is the same as in \eqref{eq:remaining1}; the term $O(N^{\delta_1(1-1/\mu)+\delta_2})$ comes from the fact that we have $O(N^{\delta_1(1-1/\mu)})$ remaining nodes at depth $\delta_1 n$, and no pruning takes place in the following $\delta_2 n-1$ steps of polarization; and the term $O(N^{(\delta_1+\delta_2)(1-1/\mu)+1-\delta_1-\delta_2})$ comes from the fact that, after pruning, there remain $O(N^{(\delta_1+\delta_2)(1-1/\mu)})$ nodes at depth $(\delta_1+\delta_2)\, n$, and each of these nodes is the root of a tree containing $2N^{1-\delta_1-\delta_2}-1$ nodes.

By doing $K$ rounds of this pruning procedure, the remaining number of nodes is given by 
\begin{equation}\label{eq:remaining3}
    O\left(\sum_{k=0}^{K-1} N^{(1-1/\mu)\sum_{j=1}^k\delta_j + \delta_{k+1} }\right).
\end{equation}
In fact, the term $O(N^{\delta_1})$ in \eqref{eq:remaining3} comes from the fact that no pruning takes places in the first $\delta_1 n-1$ steps of polarization; and, for $k\in\{1, \ldots, K-1\}$, the term $O(N^{(1-1/\mu)\sum_{j=1}^k\delta_j + \delta_{k+1} })$ comes from the fact that there are $O(N^{(1-1/\mu)\sum_{j=1}^k\delta_j})$ remaining nodes at depth $n\cdot \sum_{j=1}^k\delta_j$, and no pruning takes place in the following $\delta_{k+1} n-1$ steps of polarization.

The remaining number of nodes is given by \eqref{eq:remaining3} for any choice of $\{\delta_k\}_{k=1}^K$ such that \begin{equation}\label{eq:sumdeltas}
\sum_{k=1}^K \delta_k=1,    
\end{equation}
as the total number of polarization steps is $n$. Thus, the $\delta_k$'s are chosen in order to minimize the quantity in \eqref{eq:remaining3}. This choice requires that the exponents of $N$ in the various terms of the sum are all equal, which leads to
\begin{equation}\label{eq:deltas}
\delta_{k+1} = \delta_k/\mu, \qquad \forall \, k\in \{1, \ldots, K-1\}. 
\end{equation}
By combining \eqref{eq:deltas} with \eqref{eq:sumdeltas}, the optimal choice for the $\delta_k$'s is
\begin{equation}\label{eq:deltakfin}
    \delta_k = \frac{1}{\mu^{k-1}\displaystyle\sum_{i=0}^K\frac{1}{\mu^{i}}}.
\end{equation}
Let us emphasize that $K$ is a fixed constant which does not depend on $n$. Thus, for $k\in \{1, \ldots, K\}$, $\delta_k$ given by \eqref{eq:deltakfin} also does not depend on $n$, and there exists an integer $n_0(k)$ such that, for $n\ge n_0(k)$, the result of Lemma \ref{lemma:rate01} holds.

Consequently, the bound in \eqref{eq:remaining3} becomes
\begin{equation}\label{eq:remainingfin}
    O\left(N^{1/\sum_{i=0}^K\frac{1}{\mu^{i}}}\right),
\end{equation}
where the big-$O$ notation hides a constant that depends solely on $K$ and on $\delta_k$ for $k\in \{1, \ldots, K\}$. Also,
\begin{equation}
    \sum_{i=0}^\infty\frac{1}{\mu^{i}} = \frac{1}{1-1/\mu}.
\end{equation}
Thus, by taking $K$ sufficiently large, the number of nodes of the SSC decoding tree is $O(N^{1-1/\mu+\epsilon(K)})$, where $\epsilon(K)$ depends on $K$ and can be made arbitrarily small.

Define $\mu'$ such that $1/\mu' = 1/\mu+\epsilon(K)$, and note that the inequality \eqref{eq:suph} in Definition \ref{def:upscal} is strict. Then, for $\epsilon(K)$ sufficiently small, $\mu'$ is also an upper bound on the scaling exponent according to Definition \ref{def:upscal}, and the number of nodes of the SSC decoding tree is $O(N^{1-1/\mu'+\epsilon(K)})=O(N^{1-1/\mu})$. By using the same argument, without loss of generality, then $\mu$ is a rational number. This implies that, when $n$ is sufficiently large, $\delta_k n\in \mathbb N$ for any $k\in\{1, \ldots, K\}$, and the proof is complete. 
\end{proof}

\section{Numerical Results}\label{sec:numerical}

This section evaluates numerically the latency savings of SSC decoding relative to SC decoding, to support its sublinear latency that was proved in Theorem~\ref{th:main}. To this end, polar codes are constructed according to Definition~\ref{def:construction} and the latency $\mathcal{L}$ of the underlying decoding algorithm is calculated by counting the number of nodes in the corresponding decoding tree. Figures \ref{fig:BECres}, \ref{fig:AWGNres}, and \ref{fig:BSCres} plot the logarithm of the latency ($\log_2 \mathcal{L}$) of SC and SSC decoding as a function of $n = \log_2 N$ for $0\leq n\leq 27$. The plots consider three different families of channels: BEC in Figure~\ref{fig:BECres}, BAWGNC in Figure~\ref{fig:AWGNres}, and BSC in Figure~\ref{fig:BSCres}. For each family of channels, in the plot on the left, the channel capacity $I(W)$ is fixed to $0.5$ and the latency savings of SSC decoding is compared for two values of $p_{\rm e}$, namely, $p_{\rm e}=10^{-3}$ and $p_{\rm e}=10^{-10}$. In the plot on the right, $p_{\rm e}$ is fixed to $10^{-3}$ and the latency savings of SSC decoding is compared for three values of $I(W)$, namely, $I(W)=0.1$, $I(W)=0.5$, and $I(W)=0.9$.

The asymptotic slope of the line corresponding to the logarithm of the latency of SC decoding is $1$. In fact, the latency of SC decoding of a polar code of length $N$ is given by $2N-1$. Conversely, the asymptotic slope of the line that corresponds to the logarithm of the latency of SSC decoding is lower than $1$. Furthermore, in all the settings taken into account, this asymptotic slope is close to $1-1/\mu$. Recall that $\mu\approx 3.63$ (and $1-1/\mu\approx 0.72$) for BEC, $\mu\approx 4$ (and $1-1/\mu\approx 0.75$) for BAWGNC \cite{KMTU10}, and it is conjectured that $\mu\approx 4.2$ (and $1-1/\mu\approx 0.76$) for BSC. These asymptotic slopes are represented in the dashed blue lines in the plots. Consequently, the numerical results of Figures \ref{fig:BECres}, \ref{fig:AWGNres}, and \ref{fig:BSCres} suggest that the bound of Theorem~\ref{th:main} is tight. The latency tends to be smaller for smaller values of $p_{\rm e}$ and of $I(W)$ when the block length is not too large. However, the difference between the curves computed for different values of $p_{\rm e}$ and $I(W)$ tends to vanish as the block length increases.

\begin{figure*}[t]
\centering
\begin{subfigure}{.48\textwidth}
    \centering
    \input{figures/resultsBEC5}
    \ref{legend-BEC5}
	\caption{$I(W) = 0.5$}
	\label{fig:BECres5}
\end{subfigure}
\centering
\begin{subfigure}{.48\textwidth}
    \centering
    \input{figures/resultsBEC}
    \ref{legend-BECcomp}
	\caption{$p_{\rm e} = 10^{-3}$}
	\label{fig:BECrescomp}
\end{subfigure}
	\caption{Latency of SC and SSC decoding of polar codes constructed according to Definition~\ref{def:construction} when $W$ is a BEC.}
	\label{fig:BECres}
\end{figure*}
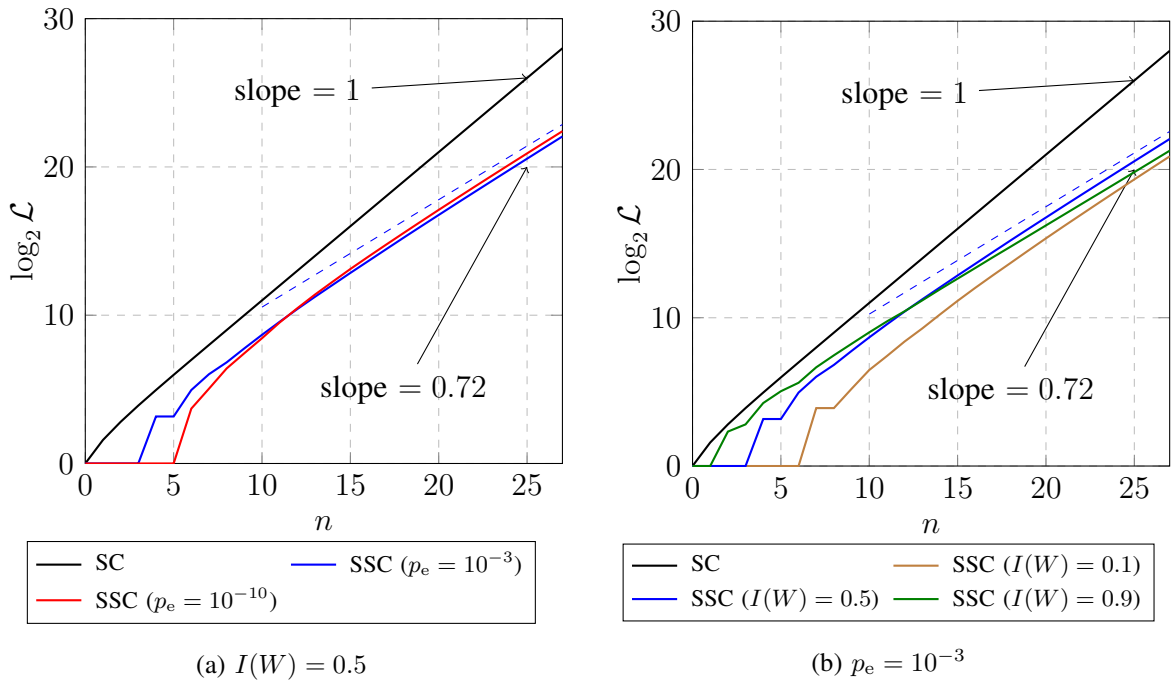

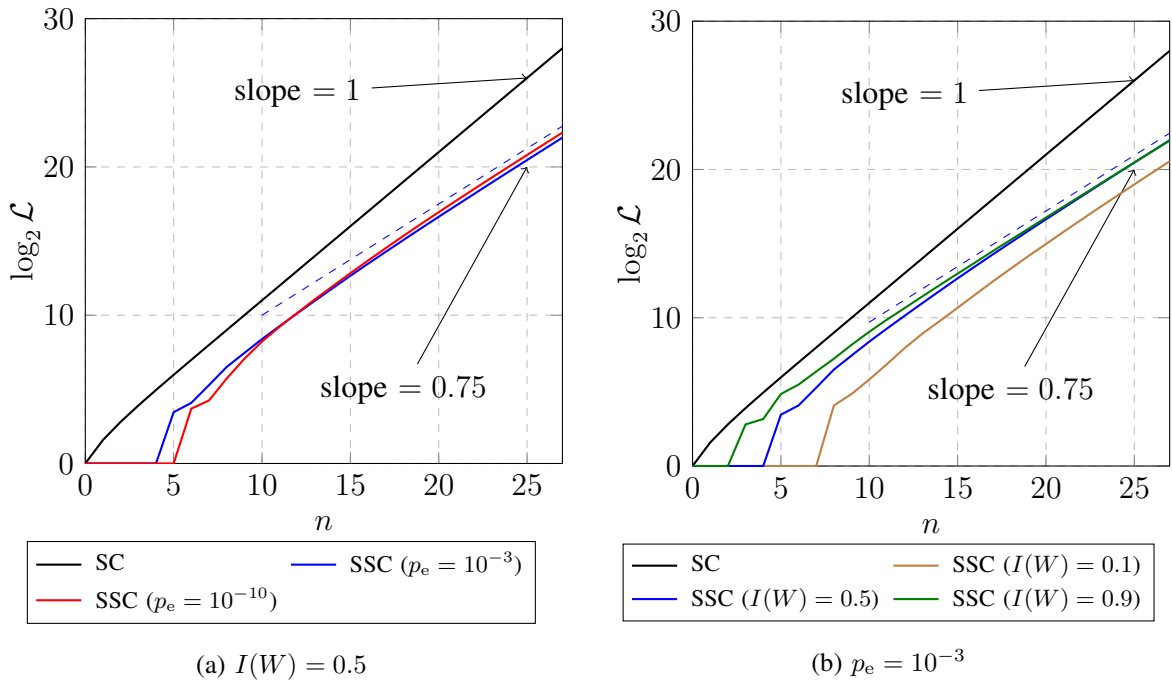
\begin{figure*}[t]
\centering
\begin{subfigure}{.48\textwidth}
    \centering
    \input{figures/resultsAWGN5}
    \ref{legend-AWGN5}
	\caption{$I(W) = 0.5$}
	\label{fig:AWGNres5}
\end{subfigure}
\centering
\begin{subfigure}{.48\textwidth}
    \centering
    \input{figures/resultsAWGN}
    \ref{legend-AWGNcomp}
	\caption{$p_{\rm e} = 10^{-3}$}
	\label{fig:AWGNrescomp}
\end{subfigure}
	\caption{Latency of SC and SSC decoding of polar codes constructed according to Definition~\ref{def:construction} when $W$ is a BAWGNC.}
	\label{fig:AWGNres}
\end{figure*}

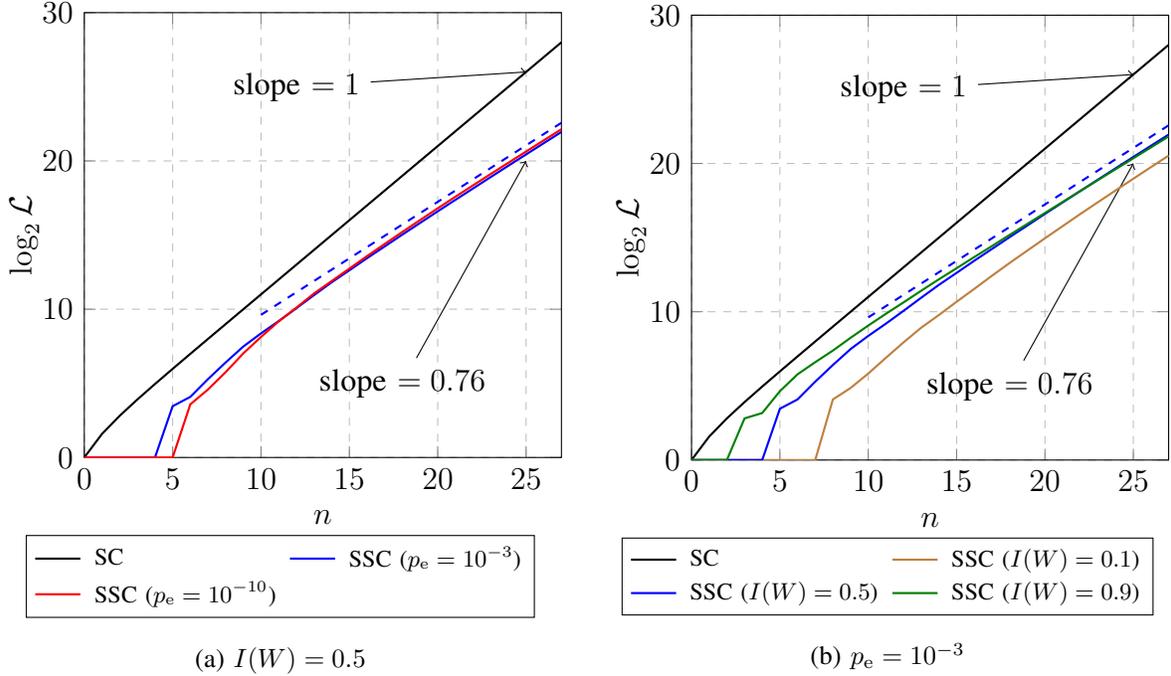
\begin{figure*}[t]
\centering
\begin{subfigure}{.48\textwidth}
    \centering
    \input{figures/resultsBSC5}
    \ref{legend-BSC5}
	\caption{$I(W) = 0.5$}
	\label{fig:BSCres5}
\end{subfigure}
\centering
\begin{subfigure}{.48\textwidth}
    \centering
    \input{figures/resultsBSC}
    \ref{legend-BSCcomp}
	\caption{$p_{\rm e} = 10^{-3}$}
	\label{fig:BSCrescomp}
\end{subfigure}
	\caption{Latency of SC and SSC decoding of polar codes constructed according to Definition~\ref{def:construction} when $W$ is a BSC.}
	\label{fig:BSCres}
\end{figure*}

The effect of using Fast-SSC decoding \cite{sarkis} is also evaluated in the numerical results. This is a pruning technique of the SC decoding tree where two additional constituent codes are introduced; these constituent codes can be decoded in parallel in a single time step. As a result, the nodes in the SC decoding tree which correspond to these two nodes can also be pruned. These additional nodes are:
\begin{itemize}
    \item \emph{Repetition (Rep) node}: A Rep node in the SC decoding tree is such that all its leaf nodes at level $0$ are frozen bits except for the rightmost leaf node, which is an information bit.
    \item \emph{Single parity-check (SPC) node}: A SPC node in the SC decoding tree is such that all its leaf nodes at level $0$ are information bits except for the leftmost leaf node, which is a frozen bit.
\end{itemize}
Figure~\ref{fig:resAll} shows the logarithm of the latency for Fast-SSC decoding in comparison with that of SC decoding and SSC decoding. Similar to the previous numerical result, this plot considers the cases in which $W$ is a BEC, a BAWGNC, and a BSC. For all three channels it can be seen that, at finite code lengths, Fast-SSC decoding brings significant latency savings compared to SSC decoding. However, the asymptotic slope of $\log_2 \mathcal{L}$ for Fast-SSC decoding is close to that of SSC decoding. Therefore, we conjecture that, asymptotically, most of the savings in latency comes from pruning Rate-0 and Rate-1 nodes.

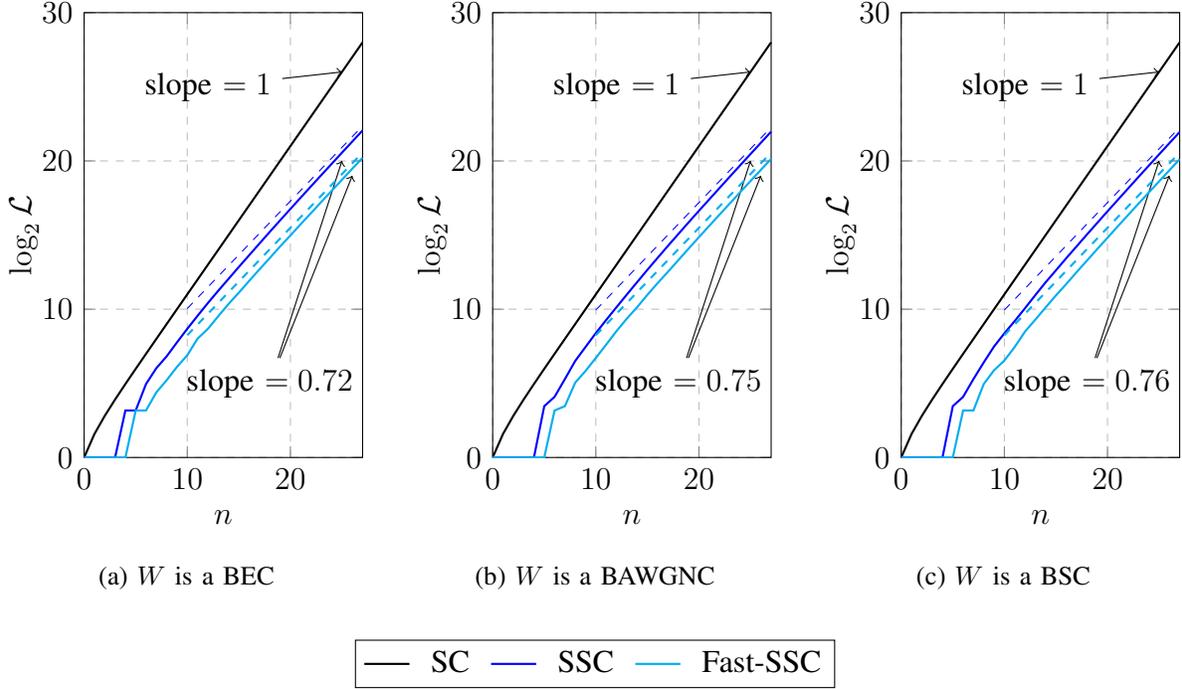
\begin{figure*}[t]
\centering
\begin{subfigure}{.32\textwidth}
    \centering
    \input{figures/resultsBECAll}
	\caption{$W$ is a BEC}
	\label{fig:BECresAll}
\end{subfigure}
\begin{subfigure}{.32\textwidth}
    \centering
    \input{figures/resultsAWGNAll}
	\caption{$W$ is a BAWGNC}
	\label{fig:AWGNresAll}
\end{subfigure}
\begin{subfigure}{.32\textwidth}
    \centering
    \input{figures/resultsBSCAll}
	\caption{$W$ is a BSC}
	\label{fig:BSCresAll}
\end{subfigure}
\ref{legend-BECAll}
	\caption{Latency of SC, SSC, and Fast-SSC decoding of polar codes constructed according to Definition~\ref{def:construction} with $p_{\rm e} = 10^{-3}$ and $I(W)=0.5$.}
	\label{fig:resAll}
\end{figure*}

\section{Conclusions}\label{sec:concl}

This paper proves that the latency of the simplified successive cancellation decoder proposed in \cite{alamdar} is sublinear in the block length. More specifically, this latency scales at most as $N^{1-1/\mu}$, where $N$ is the block length and $\mu$ is the scaling exponent of the transmission channel. This is significantly better than the latency of the standard successive cancellation decoder, which scales linearly in the block length. Numerical results show that the proposed bound is tight and that pruning additional constituent codes does not improve much the latency for large $N$. Proving rigorous lower bounds on the latency is an interesting avenue for future research. As shown in \cite{wang2019log}, changing the code construction has proved beneficial to reduce the time complexity per information bit. Thus, another interesting direction for future work is to allow for variants in the construction of the polar code, in order to further reduce the latency of the decoding process.

\section*{Acknowledgments}

M.~Mondelli was partially supported by grants NSF DMS-1613091, CCF-1714305, IIS-1741162, and ONR N00014-18-1-2729. S.~A.~Hashemi is supported by a Postdoctoral Fellowship from the Natural Sciences and Engineering Research Council of Canada (NSERC) and by Huawei. The authors would like to thank the anonymous reviewers for their comments that helped improving the quality of the manuscript.

\appendices

\section{Proof of Intermediate Lemmas}\label{app:proofs}

\begin{proof}[Proof of Lemma~\ref{lemma:unpolarized}]
We follow a strategy similar to the proof of \cite[Theorem 1]{MHU15unif-ieeeit}. Let $h(x)$ be the function of Definition \ref{def:upscal}, and define
\begin{equation}\label{eq:defrho1}
\rho_1 = \min \left( \frac{1}{2}, -\log_2\displaystyle\sup_{\substack{x\in (0, 1), y \in [x\sqrt{2-x^2}, 2x-x^2]}}\displaystyle\frac{h(x^2)+h(y)}{2h(x)}\right).
\end{equation}
Set 
\begin{equation}\label{eq:defalpha}
\gamma = \frac{1}{\nu}\log_2 \left( 1+\frac{2^{-1/\mu}-2^{-\rho_1}}{2^{-1/\mu}+2^{-\rho_1}} \right).
\end{equation}
By using \eqref{eq:suph} and the fact that $\mu > 2$, we immediately realize that $2^{-1/\mu}-2^{-\rho_1} > 0$, hence that $\gamma > 0$. In addition, it is easy to check that $\gamma <1$. Then, by \cite[Lemma 6]{MHU15unif-ieeeit}, for $n\ge 1$ and for any $\delta>0$,
\begin{equation}\label{eq:formula1}
    {\mathbb E}\left[(Z_n(1-Z_n))^{\gamma}\right] \le \frac{1}{\delta}\left(2^{-\rho_1}+\sqrt{2}\frac{\delta}{1-\delta} c_1\right)^n,
\end{equation}
where $c_1$ is a constant that depends only on $\nu$ (and not on $n, W$). Set  
\begin{equation}\label{eq:defdelta}
\delta = \frac{2^{-1/\mu}-2^{-\rho_1}}{2\sqrt{2}c_1+2^{-1/\mu}-2^{-\rho_1}}.
\end{equation}
Since $2^{-1/\mu}-2^{-\rho_1} > 0$, \eqref{eq:defdelta} is a valid choice for $\delta$. By combining \eqref{eq:formula1} and \eqref{eq:defdelta}, we deduce that
\begin{equation}\label{eq:expeqnew}
    {\mathbb E}\left[(Z_n(1-Z_n))^{\gamma}\right] \le c_2 \left(\frac{1}{2}(2^{-1/\mu}+2^{-\rho_1})\right)^n, 
\end{equation}
where $c_2$ is a constant that depends only on $\nu$ (and not on $n, W$). The proof follows from the chain of inequalities below:
\begin{equation}
\begin{split}
{\mathbb P}\left(Z_n \in \left[2^{-\nu n}, 1- 2^{-\nu n}\right]\right)
& \stackrel{\mathclap{\mbox{\footnotesize(a)}}}{=} {\mathbb P}\left( (Z_n(1-Z_n))^{\gamma}\ge (2^{-\nu n}(1- 2^{-\nu n}))^{\gamma}\right)\\
& \stackrel{\mathclap{\mbox{\footnotesize(b)}}}{\le} \frac{{\mathbb E}\left[(Z_n(1-Z_n))^{\gamma}\right]}{ (2^{-\nu n}(1- 2^{-\nu n}))^{\gamma}}\\
&\stackrel{\mathclap{\mbox{\footnotesize(c)}}}{\le} \frac{c_2 \left(\frac{1}{2}(2^{-1/\mu}+2^{-\rho_1})\right)^n}{(2^{-\nu n}(1- 2^{-\nu n}))^{\gamma}}\\
&\stackrel{\mathclap{\mbox{\footnotesize(d)}}}{\le} 2\hspace{0.1em}c_2\hspace{0.1em} \left(\frac{1}{2}(2^{-1/\mu}+2^{-\rho_1})2^{\nu\gamma}\right)^n \stackrel{\mathclap{\mbox{\footnotesize(e)}}}{=} 2\hspace{0.1em}c_2\hspace{0.1em} 2^{-n/\mu},
\end{split}
\end{equation}
where the equality (a) uses the concavity of the function $f(x)=(x(1-x))^\gamma$; the inequality (b) follows from Markov's inequality; the inequality (c) uses \eqref{eq:expeqnew}; the inequality (d) uses that $1-2^{-\nu n}\ge 1/2$ for any $n\ge 1$ and $\nu>1$; and the equality (e) uses the definition \eqref{eq:defalpha}.
\end{proof}

\begin{remark}[Number of un-polarized channels for BEC]\label{rmk:BEC}
For the special case in which $W$ is a BEC, the result \eqref{eq:nummid} holds, where $\mu$ is an upper bound on the scaling exponent of BEC according to Definition \ref{def:ubscalBEC}. This is proved by setting
\begin{equation}\label{eq:defrho1BEC}
\rho_1 = \min \left( \frac{1}{2}, -\log_2\displaystyle\sup_{x\in (0, 1)}\displaystyle\frac{h(x^2)+h(2x-x^2)}{2h(x)}\right),
\end{equation}
and by following the same argument of the proof above.
\end{remark}

\begin{proof}[Proof of Lemma~\ref{lemma:rate01}]
There are two scenarios in Lemma~\ref{lemma:rate01}: $Z(W)\le 1/N^3$ and $Z(W)\ge 1-1/N^3$. We start with the first scenario, $Z(W)\le 1/N^3$. Note that \eqref{eq:eqBMSC} implies that $Z_n\le 2Z_{n-1}$. Thus, as $Z(W)\le 1/N^3$, for any $i\in\{1, \ldots, N/M\}$, we have that
\begin{equation}\label{eq:ubBattalemma1}
    Z(W_{n-m}^{(i)})\le \frac{2^{n-m}}{N^3}=\frac{1}{M\cdot N^2}\le \frac{1}{N^2}.
\end{equation}
Consequently, for sufficiently large $N$, $Z(W_{n-m}^{(i)})\le p_{\rm e}/ N$ for any $i$, and $\mathcal C_{\rm polar}(p_{\rm e}/M, W, N/M)$ has rate $1$.

Let us now consider the second scenario, $Z(W)\ge 1-1/N^3$. Consider the random process $1-Z_n$ and note that \eqref{eq:eqBMSC} implies that $1-Z_n\le 2(1-Z_{n-1})$. Thus, as $1-Z(W)\le 1/N^3$, for any $i\in\{1, \ldots, N/M\}$, we have that 
\begin{equation}
  1-  Z(W_{n-m}^{(i)})\le \frac{2^{n-m}}{N^3}=\frac{1}{M\cdot N^2}\le \frac{1}{N^2}.
\end{equation}
Consequently, for sufficiently large $N$, $Z(W_{n-m}^{(i)})> p_{\rm e}/N$ for any $i$, and $\mathcal C_{\rm polar}(p_{\rm e}/M, W, N/M)$ has rate $0$.
\end{proof}

\bibliographystyle{IEEEtran}
\bibliography{lth,lthpub}

\end{document}

%% file: figures/sc-dec.tex
\begin{tikzpicture}[scale=1.9, thick]
  \draw [fill=lightgray] (0,0) circle [radius=.05];

  \draw [fill=lightgray] (-1,-.5) circle [radius=.05];
  \draw [fill=lightgray] (1,-.5) circle [radius=.05];

  \draw (-1.5,-1) circle [radius=.05];
  \draw [fill=lightgray] (-.5,-1) circle [radius=.05];
  \draw [fill=lightgray] (.5,-1) circle [radius=.05];
  \draw [fill=black] (1.5,-1) circle [radius=.05];

  \draw (-1.75,-1.5) circle [radius=.05];
  \draw (-1.25,-1.5) circle [radius=.05];
  \draw (-.75,-1.5) circle [radius=.05];
  \draw [fill=black] (-.25,-1.5) circle [radius=.05];
  \draw (.25,-1.5) circle [radius=.05];
  \draw [fill=black] (.75,-1.5) circle [radius=.05];
  \draw [fill=black] (1.25,-1.5) circle [radius=.05];
  \draw [fill=black] (1.75,-1.5) circle [radius=.05];

  \node at (-1.75,-1.7) {$\hat{u}_1$};
  \node at (-1.25,-1.7) {$\hat{u}_2$};
  \node at (-.75,-1.7) {$\hat{u}_3$};
  \node at (-.25,-1.7) {$\hat{u}_4$};
  \node at (.25,-1.7) {$\hat{u}_5$};
  \node at (.75,-1.7) {$\hat{u}_6$};
  \node at (1.25,-1.7) {$\hat{u}_7$};
  \node at (1.75,-1.7) {$\hat{u}_8$};

  \draw (0,-.05) -- (-1,-.45);
  \draw (0,-.05) -- (1,-.45);

  \draw (-1,-.55) -- (-1.5,-.95);
  \draw (-1,-.55) -- (-.5,-.95);
  \draw (1,-.55) -- (.5,-.95);
  \draw (1,-.55) -- (1.5,-.95);

  \draw (-1.5,-1.05) -- (-1.75,-1.45);
  \draw (-1.5,-1.05) -- (-1.25,-1.45);
  \draw (-.5,-1.05) -- (-.75,-1.45);
  \draw (-.5,-1.05) -- (-.25,-1.45);
  \draw (.5,-1.05) -- (.25,-1.45);
  \draw (.5,-1.05) -- (.75,-1.45);
  \draw (1.5,-1.05) -- (1.25,-1.45);
  \draw (1.5,-1.05) -- (1.75,-1.45);

  \draw [thin,dotted] (-2,0) -- (2,0);
  \draw [thin,dotted] (-2,-.5) -- (2,-.5);
  \draw [thin,dotted] (-2,-1) -- (2,-1);
  \draw [thin,dotted] (-2,-1.5) -- (2,-1.5);

  \node at (-2.25,0) {$s=3$};
  \node at (-2.25,-.5) {$s=2$};
  \node at (-2.25,-1) {$s=1$};
  \node at (-2.25,-1.5) {$s=0$};

  \draw [->] (-.12,-.05) -- (-1,-.4) node [above=-.1cm,midway,rotate=25] {$\bm{\alpha}$};
  \draw [->] (-.88,-.45) -- (0,-.1) node [below=-.1cm,midway,rotate=25] {$\bm{\beta}$};

  \draw [->] (-1.06,-.55) -- (-1.5,-.9) node [above=-.1cm,midway,rotate=40] {$\bm{\alpha}^{\ell}$};
  \draw [->] (-1.44,-.95) -- (-1.0,-0.6) node [below=-.1cm,midway,rotate=40] {$\bm{\beta}^{\ell}$};

  \draw [<-] (-.94,-.55) -- (-.5,-.9) node [above=-.1cm,midway,rotate=-40] {$\bm{\beta}^{\text{r}}$};
  \draw [<-] (-.56,-.95) -- (-0.975,-.625) node [below=-.1cm,midway,rotate=-40] {$\bm{\alpha}^{\text{r}}$};

\end{tikzpicture}

%% file: figures/ssc-dec.tex
\begin{tikzpicture}[scale=1.9, thick]
  \draw [fill=lightgray] (0,0) circle [radius=.05];

  \draw [fill=lightgray] (-1,-.5) circle [radius=.05];
  \draw [fill=lightgray] (1,-.5) circle [radius=.05];

  \draw (-1.5,-1) circle [radius=.05];
  \draw [fill=lightgray] (-.5,-1) circle [radius=.05];
  \draw [fill=lightgray] (.5,-1) circle [radius=.05];
  \draw [fill=black] (1.5,-1) circle [radius=.05];

  \draw (-.75,-1.5) circle [radius=.05];
  \draw [fill=black] (-.25,-1.5) circle [radius=.05];
  \draw (.25,-1.5) circle [radius=.05];
  \draw [fill=black] (.75,-1.5) circle [radius=.05];

  \draw (0,-.05) -- (-1,-.45);
  \draw (0,-.05) -- (1,-.45);

  \draw (-1,-.55) -- (-1.5,-.95);
  \draw (-1,-.55) -- (-.5,-.95);
  \draw (1,-.55) -- (.5,-.95);
  \draw (1,-.55) -- (1.5,-.95);

  \draw (-.5,-1.05) -- (-.75,-1.45);
  \draw (-.5,-1.05) -- (-.25,-1.45);
  \draw (.5,-1.05) -- (.25,-1.45);
  \draw (.5,-1.05) -- (.75,-1.45);

  \draw [thin,dotted] (-2,0) -- (2,0);
  \draw [thin,dotted] (-2,-.5) -- (2,-.5);
  \draw [thin,dotted] (-2,-1) -- (2,-1);
  \draw [thin,dotted] (-2,-1.5) -- (2,-1.5);

  \node at (-2.25,0) {$s=3$};
  \node at (-2.25,-.5) {$s=2$};
  \node at (-2.25,-1) {$s=1$};
  \node at (-2.25,-1.5) {$s=0$};
\end{tikzpicture}

%% file: figures/resultsBEC5.tex
\begin{tikzpicture}

\begin{axis}[
scale=1,
xmin=0,
xmax=27,
ymin=0,
ymax=30,
grid=both,
ymajorgrids=true,
xmajorgrids=true,
grid style=dashed,
width=\textwidth, height=7.5cm,
xlabel={$n$},
ylabel={$\log_2 \mathcal{L}$},
ylabel shift=-7,
legend cell align={left},
legend pos=north west,
legend style={
	column sep= 1mm,
	font=\fontsize{9pt}{9}\selectfont,
},
legend to name=legend-BEC5,
legend columns=2,
]

\node (comment1) at (axis cs:12,25){slope $=1$};
\draw[->] (comment1) -- (axis cs:25,26);

\node (comment) at (axis cs:18,5){slope $=0.72$};
\draw[->] (comment) -- (axis cs:25,20);

\addplot[
color=black,
thick,
]
table {
0  0
1  1.584963
2  2.807355
3  3.906891
4  4.954196
5  5.977280
6  6.988685
7  7.994353
8  8.997179
9  9.998590
10 10.999295
11 11.999648
12 12.999824
13 13.999912
14 14.999956
15 15.999978
16 16.999989
17 17.999994
18 18.999997
19 19.999999
20 20.999999
21 22.000000
22 23.000000
23 24.000000
24 25.000000
25 26.000000
26 27.000000
27 28.000000
};
\addlegendentry{SC}

\addplot[
color=blue,
thick,
]
table {
0  0
1  0
2  0
3  0
4  3.169925001442312151e+00
5  3.169925001442312151e+00
6  4.954196310386874913e+00
7  6.022367813028454364e+00
8  6.820178962415187840e+00
9  7.761551232444479531e+00
10 8.675957032941749247e+00
11 9.556506054671928041e+00
12 1.041679752760606092e+01
13 1.124257868945134575e+01
14 1.204678297035635026e+01
15 1.285155385842977793e+01
16 1.363945376809673782e+01
17 1.442999510300424681e+01
18 1.521663219602045913e+01
19 1.599189823692343637e+01
20 1.676101874564537653e+01
21 1.752772897918726969e+01
22 1.828869938626520764e+01
23 1.904654053306329331e+01
24 1.980145734299838978e+01
25 2.055595680427493477e+01
26 2.130626699765168297e+01
27 2.205552903495953387e+01
};
\addlegendentry{SSC ($p_{\rm e} = 10^{-3}$)}

\addplot[
color=red,
thick,
]
table {
0  0
1  0
2  0
3  0
4  0
5  0
6  3.700439718141092182e+00
7  5.044394119358453388e+00
8  6.409390936137701722e+00
9  7.434628227636724596e+00
10 8.447083226209652906e+00
11 9.509775004326936454e+00
12 1.047065887406055218e+01
13 1.139285403987287282e+01
14 1.226062555606822713e+01
15 1.311618148643118253e+01
16 1.393378284877537254e+01
17 1.473888391977510004e+01
18 1.553390761937179754e+01
19 1.632268517021771359e+01
20 1.710366687591252699e+01
21 1.787316170769939561e+01
22 1.864022728095296344e+01
23 1.940212313449090686e+01
24 2.015886728228723612e+01
25 2.091170138199329998e+01
26 2.166153500168453050e+01
27 2.240851037527678358e+01
};
\addlegendentry{SSC ($p_{\rm e} = 10^{-10}$)}

\addplot[
color=blue,
dashed,
domain=10:27,
samples=10,
smooth,
]
{(1-1/3.63)*x+3.3};

\end{axis}
\end{tikzpicture}

%% file: figures/resultsBEC.tex
\begin{tikzpicture}

\begin{axis}[
scale=1,
xmin=0,
xmax=27,
ymin=0,
ymax=30,
grid=both,
ymajorgrids=true,
xmajorgrids=true,
grid style=dashed,
width=\textwidth, height=7.5cm,
xlabel={$n$},
ylabel={$\log_2 \mathcal{L}$},
ylabel shift=-7,
legend cell align={left},
legend pos=north west,
legend style={
	column sep= 1mm,
	font=\fontsize{9pt}{9}\selectfont,
},
legend to name=legend-BECcomp,
legend columns=2,
]

\node (comment1) at (axis cs:12,25){slope $=1$};
\draw[->] (comment1) -- (axis cs:25,26);

\node (comment) at (axis cs:18,5){slope $=0.72$};
\draw[->] (comment) -- (axis cs:25,20);

\addplot[
color=black,
thick,
]
table {
0  0
1  1.584963
2  2.807355
3  3.906891
4  4.954196
5  5.977280
6  6.988685
7  7.994353
8  8.997179
9  9.998590
10 10.999295
11 11.999648
12 12.999824
13 13.999912
14 14.999956
15 15.999978
16 16.999989
17 17.999994
18 18.999997
19 19.999999
20 20.999999
21 22.000000
22 23.000000
23 24.000000
24 25.000000
25 26.000000
26 27.000000
27 28.000000
};
\addlegendentry{SC}

\addplot[
color=brown,
thick,
]
table {
0  0
1  0
2  0
3  0
4  0
5  0
6  0
7  3.906890595608518701e+00
8  3.906890595608518701e+00
9  5.209453365628950117e+00
10 6.475733430966397641e+00
11 7.417852514885898252e+00
12 8.396604781181858712e+00
13 9.278449458220482171e+00
14 1.022037832769522758e+01
15 1.114019070319183591e+01
16 1.201854794187165076e+01
17 1.286708564661683596e+01
18 1.369642424057887808e+01
19 1.453192967793801138e+01
20 1.535572961655076973e+01
21 1.616571177748645738e+01
22 1.696286212401254190e+01
23 1.775920551840231099e+01
24 1.854640842679112467e+01
25 1.933132592091497415e+01
26 2.010794879327808360e+01
27 2.088174953564698200e+01
};
\addlegendentry{SSC ($I(W) = 0.1$)}

\addplot[
color=blue,
thick,
]
table {
0  0
1  0
2  0
3  0
4  3.169925001442312151e+00
5  3.169925001442312151e+00
6  4.954196310386874913e+00
7  6.022367813028454364e+00
8  6.820178962415187840e+00
9  7.761551232444479531e+00
10 8.675957032941749247e+00
11 9.556506054671928041e+00
12 1.041679752760606092e+01
13 1.124257868945134575e+01
14 1.204678297035635026e+01
15 1.285155385842977793e+01
16 1.363945376809673782e+01
17 1.442999510300424681e+01
18 1.521663219602045913e+01
19 1.599189823692343637e+01
20 1.676101874564537653e+01
21 1.752772897918726969e+01
22 1.828869938626520764e+01
23 1.904654053306329331e+01
24 1.980145734299838978e+01
25 2.055595680427493477e+01
26 2.130626699765168297e+01
27 2.205552903495953387e+01
};
\addlegendentry{SSC ($I(W) = 0.5$)}

\addplot[
color=green!50!black,
thick,
]
table {
0  0
1  0
2  2.321928094887362182e+00
3  2.807354922057604174e+00
4  4.247927513443585212e+00
5  5.044394119358453388e+00
6  5.614709844115208348e+00
7  6.658211482751794641e+00
8  7.467605550082997645e+00
9  8.243173983472951605e+00
10 9.014020470314934030e+00
11 9.753216749178955425e+00
12 1.043567026093655059e+01
13 1.117180228868298286e+01
14 1.191699904908280772e+01
15 1.263639804888383722e+01
16 1.335631440692812077e+01
17 1.408339622038513816e+01
18 1.479365410750623511e+01
19 1.551067146263640417e+01
20 1.622755954003473633e+01
21 1.694339621881450952e+01
22 1.766225808137203046e+01
23 1.837834142976758756e+01
24 1.909666111417160295e+01
25 1.981940106240384836e+01
26 2.054247174711013102e+01
27 2.126616487487542884e+01
};
\addlegendentry{SSC ($I(W) = 0.9$)}

\addplot[
color=blue,
dashed,
domain=10:27,
samples=10,
smooth,
]
{(1-1/3.63)*x+3};

\end{axis}
\end{tikzpicture}

%% file: figures/resultsAWGN5.tex
\begin{tikzpicture}

\begin{axis}[
scale=1,
xmin=0,
xmax=27,
ymin=0,
ymax=30,
grid=both,
ymajorgrids=true,
xmajorgrids=true,
grid style=dashed,
width=\textwidth, height=7.5cm,
xlabel={$n$},
ylabel={$\log_2 \mathcal{L}$},
ylabel shift=-7,
legend cell align={left},
legend pos=north west,
legend style={
	column sep= 1mm,
	font=\fontsize{9pt}{9}\selectfont,
},
legend to name=legend-AWGN5,
legend columns=2,
]

\node (comment1) at (axis cs:12,25){slope $=1$};
\draw[->] (comment1) -- (axis cs:25,26);

\node (comment) at (axis cs:18,5){slope $=0.75$};
\draw[->] (comment) -- (axis cs:25,20);

\addplot[
color=black,
thick,
]
table {
0  0
1  1.584963
2  2.807355
3  3.906891
4  4.954196
5  5.977280
6  6.988685
7  7.994353
8  8.997179
9  9.998590
10 10.999295
11 11.999648
12 12.999824
13 13.999912
14 14.999956
15 15.999978
16 16.999989
17 17.999994
18 18.999997
19 19.999999
20 20.999999
21 22.000000
22 23.000000
23 24.000000
24 25.000000
25 26.000000
26 27.000000
27 28.000000
};
\addlegendentry{SC}

\addplot[
color=blue,
thick,
]
table {
0  0
1  0
2  0
3  0
4  0
5  3.459431618637297312e+00
6  4.087462841250339096e+00
7  5.285402218862248702e+00
8  6.507794640198696356e+00
9  7.451211111832328982e+00
10 8.388017285345135576e+00
11 9.269126679149417569e+00
12 1.013057056280542589e+01
13 1.098228060455828370e+01
14 1.182137539237253954e+01
15 1.265664831174735028e+01
16 1.347078595576430082e+01
17 1.427663333088670328e+01
18 1.506865233891319278e+01
19 1.585941345668450708e+01
20 1.663703519100529604e+01
21 1.740647907283767992e+01
22 1.817561289715512629e+01
23 1.894365089933583945e+01
24 1.970479356355991385e+01
25 2.046368899950802600e+01
26 2.122038009912986212e+01
27 2.197507091502399845e+01
};
\addlegendentry{SSC ($p_{\rm e} = 10^{-3}$)}

\addplot[
color=red,
thick,
]
table {
0  0
1  0
2  0
3  0
4  0
5  0
6  3.700439718141092182e+00
7  4.247927513443585212e+00
8  5.727920454563198760e+00
9  7.055282435501189831e+00
10 8.224001674198104794e+00
11 9.226412192788785660e+00
12 1.015860968821447763e+01
13 1.107748335685950813e+01
14 1.195383268214528094e+01
15 1.283229667874988422e+01
16 1.368266546424580987e+01
17 1.452399068248440983e+01
18 1.534634058829055725e+01
19 1.615384886828920585e+01
20 1.694639230502695071e+01
21 1.773263780453225280e+01
22 1.850880450142529199e+01
23 1.927982138787777444e+01
24 2.004473962760934924e+01
25 2.080511957881320129e+01
26 2.156194468334127023e+01
27 2.231595137360994840e+01
};
\addlegendentry{SSC ($p_{\rm e} = 10^{-10}$)}

\addplot[
color=blue,
dashed,
domain=10:27,
samples=10,
smooth,
]
{(1-1/4.0)*x+2.5};

\end{axis}
\end{tikzpicture}

%% file: figures/resultsAWGN.tex
\begin{tikzpicture}

\begin{axis}[
scale=1,
xmin=0,
xmax=27,
ymin=0,
ymax=30,
grid=both,
ymajorgrids=true,
xmajorgrids=true,
grid style=dashed,
width=\textwidth, height=7.5cm,
xlabel={$n$},
ylabel={$\log_2 \mathcal{L}$},
ylabel shift=-7,
legend cell align={left},
legend pos=north west,
legend style={
	column sep= 1mm,
	font=\fontsize{9pt}{9}\selectfont,
},
legend to name=legend-AWGNcomp,
legend columns=2,
]

\node (comment1) at (axis cs:12,25){slope $=1$};
\draw[->] (comment1) -- (axis cs:25,26);

\node (comment) at (axis cs:18,5){slope $=0.75$};
\draw[->] (comment) -- (axis cs:25,20);

\addplot[
color=black,
thick,
]
table {
0  0
1  1.584963
2  2.807355
3  3.906891
4  4.954196
5  5.977280
6  6.988685
7  7.994353
8  8.997179
9  9.998590
10 10.999295
11 11.999648
12 12.999824
13 13.999912
14 14.999956
15 15.999978
16 16.999989
17 17.999994
18 18.999997
19 19.999999
20 20.999999
21 22.000000
22 23.000000
23 24.000000
24 25.000000
25 26.000000
26 27.000000
27 28.000000
};
\addlegendentry{SC}

\addplot[
color=brown,
thick,
]
table {
0  0
1  0
2  0
3  0
4  0
5  0
6  0
7  0
8  4.087462841250339096e+00
9  4.857980995127571866e+00
10 5.832890014164741288e+00
11 6.870364719583404778e+00
12 7.948367231584677839e+00
13 8.921840937074490441e+00
14 9.789533644970360271e+00
15 1.067154106671188352e+01
16 1.156843119316745927e+01
17 1.245301335446631441e+01
18 1.330791230249294799e+01
19 1.415473901176099147e+01
20 1.497284520002758335e+01
21 1.580011627300073584e+01
22 1.660699792384004425e+01
23 1.741304111988607062e+01
24 1.820588391125845718e+01
25 1.899346639738936204e+01
26 1.977548007587450840e+01
27 2.055200389940937455e+01
};
\addlegendentry{SSC ($I(W) = 0.1$)}

\addplot[
color=blue,
thick,
]
table {
0  0
1  0
2  0
3  0
4  0
5  3.459431618637297312e+00
6  4.087462841250339096e+00
7  5.285402218862248702e+00
8  6.507794640198696356e+00
9  7.451211111832328982e+00
10 8.388017285345135576e+00
11 9.269126679149417569e+00
12 1.013057056280542589e+01
13 1.098228060455828370e+01
14 1.182137539237253954e+01
15 1.265664831174735028e+01
16 1.347078595576430082e+01
17 1.427663333088670328e+01
18 1.506865233891319278e+01
19 1.585941345668450708e+01
20 1.663703519100529604e+01
21 1.740647907283767992e+01
22 1.817561289715512629e+01
23 1.894365089933583945e+01
24 1.970479356355991385e+01
25 2.046368899950802600e+01
26 2.122038009912986212e+01
27 2.197507091502399845e+01
};
\addlegendentry{SSC ($I(W) = 0.5$)}

\addplot[
color=green!50!black,
thick,
]
table {
0  0
1  0
2  0
3  2.807354922057604174e+00
4  3.169925001442312151e+00
5  4.857980995127571866e+00
6  5.491853096329674777e+00
7  6.375039431346924523e+00
8  7.238404739325078552e+00
9  8.174925682500678192e+00
10 9.052568050804152833e+00
11 9.871905237659186483e+00
12 1.064835758200966609e+01
13 1.143619099548141627e+01
14 1.220914871035574656e+01
15 1.297351800470371685e+01
16 1.373290936193115819e+01
17 1.450140182330193284e+01
18 1.525100153076417620e+01
19 1.600244146250286548e+01
20 1.675450914167941363e+01
21 1.749982253526091469e+01
22 1.824346688957258777e+01
23 1.898712797568486366e+01
24 1.972801847485433413e+01
25 2.046888948882870451e+01
26 2.120763697832870776e+01
27 2.194637052866776727e+01
};
\addlegendentry{SSC ($I(W) = 0.9$)}

\addplot[
color=blue,
dashed,
domain=10:27,
samples=10,
smooth,
]
{(1-1/4.0)*x+2.2};

\end{axis}
\end{tikzpicture}

%% file: figures/resultsBSC5.tex
\begin{tikzpicture}

\begin{axis}[
scale=1,
xmin=0,
xmax=27,
ymin=0,
ymax=30,
grid=both,
ymajorgrids=true,
xmajorgrids=true,
grid style=dashed,
width=\textwidth, height=7.5cm,
xlabel={$n$},
ylabel={$\log_2 \mathcal{L}$},
ylabel shift=-7,
legend cell align={left},
legend pos=north west,
legend style={
	column sep= 1mm,
	font=\fontsize{9pt}{9}\selectfont,
},
legend to name=legend-BSC5,
legend columns=2,
]

\node (comment1) at (axis cs:12,25){slope $=1$};
\draw[->] (comment1) -- (axis cs:25,26);

\node (comment) at (axis cs:18,5){slope $=0.76$};
\draw[->] (comment) -- (axis cs:25,20);

\addplot[
color=black,
thick,
]
table {
0  0
1  1.584963
2  2.807355
3  3.906891
4  4.954196
5  5.977280
6  6.988685
7  7.994353
8  8.997179
9  9.998590
10 10.999295
11 11.999648
12 12.999824
13 13.999912
14 14.999956
15 15.999978
16 16.999989
17 17.999994
18 18.999997
19 19.999999
20 20.999999
21 22.000000
22 23.000000
23 24.000000
24 25.000000
25 26.000000
26 27.000000
27 28.000000
};
\addlegendentry{SC}

\addplot[
color=blue,
thick,
]
table {
0  0
1  0
2  0
3  0
4  0
5  3.459431618637297312e+00
6  4.087462841250339096e+00
7  5.285402218862248702e+00
8  6.409390936137701722e+00
9  7.483815777264256397e+00
10 8.370687406807217457e+00
11 9.197216693110052077e+00
12 1.006474276475025675e+01
13 1.094909715572920739e+01
14 1.181498293642679087e+01
15 1.262182294739420207e+01
16 1.344047989881121907e+01
17 1.423369456554691226e+01
18 1.503217490229382491e+01
19 1.581545835083124629e+01
20 1.660074050438925752e+01
21 1.737353640716244740e+01
22 1.814266052948796215e+01
23 1.891275779561843606e+01
24 1.967743921421040554e+01
25 2.043903848732817252e+01
26 2.119554943325594465e+01
27 2.195090576136351146e+01
};
\addlegendentry{SSC ($p_{\rm e} = 10^{-3}$)}

\addplot[
color=red,
thick,
]
table {
0  0
1  0
2  0
3  0
4  0
5  0
6  3.584962500721156076e+00
7  4.584962500721156076e+00
8  5.754887502163468227e+00
9  7.022367813028454364e+00
10 8.129283016944967244e+00
11 9.179909090014934492e+00
12 1.010066233900519883e+01
13 1.106069593168755461e+01
14 1.189860138740393758e+01
15 1.273682469909886450e+01
16 1.358167135857814323e+01
17 1.439807562290986809e+01
18 1.521219237190847196e+01
19 1.601227534485612125e+01
20 1.680367708178865982e+01
21 1.758369982497557515e+01
22 1.835381018308623169e+01
23 1.911841391269988932e+01
24 1.987774544797731124e+01
25 2.063334293967313471e+01
26 2.138677714671722896e+01
27 2.213748526220379276e+01
};
\addlegendentry{SSC ($p_{\rm e} = 10^{-10}$)}

\addplot[
color=blue,
dashed,
thick,
domain=10:27,
samples=10,
smooth,
]
{(1-1/4.2)*x+2};

\end{axis}
\end{tikzpicture}

%% file: figures/resultsBSC.tex
\begin{tikzpicture}

\begin{axis}[
scale=1,
xmin=0,
xmax=27,
ymin=0,
ymax=30,
grid=both,
ymajorgrids=true,
xmajorgrids=true,
grid style=dashed,
width=\textwidth, height=7.5cm,
xlabel={$n$},
ylabel={$\log_2 \mathcal{L}$},
ylabel shift=-7,
legend cell align={left},
legend pos=north west,
legend style={
	column sep= 1mm,
	font=\fontsize{9pt}{9}\selectfont,
},
legend to name=legend-BSCcomp,
legend columns=2,
]

\node (comment1) at (axis cs:12,25){slope $=1$};
\draw[->] (comment1) -- (axis cs:25,26);

\node (comment) at (axis cs:18,5){slope $=0.76$};
\draw[->] (comment) -- (axis cs:25,20);

\addplot[
color=black,
thick,
]
table {
0  0
1  1.584963
2  2.807355
3  3.906891
4  4.954196
5  5.977280
6  6.988685
7  7.994353
8  8.997179
9  9.998590
10 10.999295
11 11.999648
12 12.999824
13 13.999912
14 14.999956
15 15.999978
16 16.999989
17 17.999994
18 18.999997
19 19.999999
20 20.999999
21 22.000000
22 23.000000
23 24.000000
24 25.000000
25 26.000000
26 27.000000
27 28.000000
};
\addlegendentry{SC}

\addplot[
color=brown,
thick,
]
table {
0  0
1  0
2  0
3  0
4  0
5  0
6  0
7  0
8  4.087462841250339096e+00
9  4.857980995127571866e+00
10 5.832890014164741288e+00
11 6.894817763307943714e+00
12 7.936637939002570974e+00
13 8.927777962082341645e+00
14 9.782998208920414385e+00
15 1.067507492038544470e+01
16 1.154254826233521314e+01
17 1.242862188609111129e+01
18 1.329073886243934766e+01
19 1.412468677365149539e+01
20 1.495587689984139423e+01
21 1.577398911590227470e+01
22 1.659305550542027774e+01
23 1.739034531811393336e+01
24 1.818395031638913295e+01
25 1.897036713518178530e+01
26 1.975371785187162388e+01
27 2.053056350602889424e+01
};
\addlegendentry{SSC ($I(W) = 0.1$)}

\addplot[
color=blue,
thick,
]
table {
0  0
1  0
2  0
3  0
4  0
5  3.459431618637297312e+00
6  4.087462841250339096e+00
7  5.285402218862248702e+00
8  6.409390936137701722e+00
9  7.483815777264256397e+00
10 8.370687406807217457e+00
11 9.197216693110052077e+00
12 1.006474276475025675e+01
13 1.094909715572920739e+01
14 1.181498293642679087e+01
15 1.262182294739420207e+01
16 1.344047989881121907e+01
17 1.423369456554691226e+01
18 1.503217490229382491e+01
19 1.581545835083124629e+01
20 1.660074050438925752e+01
21 1.737353640716244740e+01
22 1.814266052948796215e+01
23 1.891275779561843606e+01
24 1.967743921421040554e+01
25 2.043903848732817252e+01
26 2.119554943325594465e+01
27 2.195090576136351146e+01
};
\addlegendentry{SSC ($I(W) = 0.5$)}

\addplot[
color=green!50!black,
thick,
]
table {
0  0
1  0
2  0
3  2.807354922057604174e+00
4  3.169925001442312151e+00
5  4.643856189774724363e+00
6  5.781359713524659938e+00
7  6.599912842187127815e+00
8  7.383704292474051911e+00
9  8.243173983472951605e+00
10 9.068778277985412473e+00
11 9.862637357558794449e+00
12 1.063390340934851963e+01
13 1.141943355070537791e+01
14 1.218951572980019371e+01
15 1.293571755935994183e+01
16 1.369054368199859262e+01
17 1.442907977030915134e+01
18 1.518351908043128518e+01
19 1.592321140397268131e+01
20 1.666396108195772285e+01
21 1.740406881809547812e+01
22 1.814316893841166234e+01
23 1.888159496209824795e+01
24 1.961663678404383049e+01
25 2.035163913755190279e+01
26 2.108611095047603712e+01
27 2.182209191834190065e+01
};
\addlegendentry{SSC ($I(W) = 0.9$)}

\addplot[
color=blue,
dashed,
thick,
domain=10:27,
samples=10,
smooth,
]
{(1-1/4.2)*x+2};

\end{axis}
\end{tikzpicture}

%% file: figures/resultsBECAll.tex
\begin{tikzpicture}

\begin{axis}[
scale=1,
xmin=0,
xmax=27,
ymin=0,
ymax=30,
grid=both,
ymajorgrids=true,
xmajorgrids=true,
grid style=dashed,
width=\textwidth, height=7.5cm,
xlabel={$n$},
ylabel={$\log_2 \mathcal{L}$},
ylabel shift=-7,
legend cell align={left},
legend pos=north west,
legend style={
	column sep= 2mm,
},
legend to name=legend-BECAll,
legend columns=3,
]

\node (comment1) at (axis cs:12,25){slope $=1$};
\draw[->] (comment1) -- (axis cs:25,26);

\node (comment) at (axis cs:18,5){slope $=0.72$};
\draw[->] (comment) -- (axis cs:25,20);
\draw[->] (comment) -- (axis cs:26,19);

\addplot[
color=black,
thick,
]
table {
0  0
1  1.584963
2  2.807355
3  3.906891
4  4.954196
5  5.977280
6  6.988685
7  7.994353
8  8.997179
9  9.998590
10 10.999295
11 11.999648
12 12.999824
13 13.999912
14 14.999956
15 15.999978
16 16.999989
17 17.999994
18 18.999997
19 19.999999
20 20.999999
21 22.000000
22 23.000000
23 24.000000
24 25.000000
25 26.000000
26 27.000000
27 28.000000
};
\addlegendentry{SC}

\addplot[
color=blue,
thick,
]
table {
0  0
1  0
2  0
3  0
4  3.169925001442312151e+00
5  3.169925001442312151e+00
6  4.954196310386874913e+00
7  6.022367813028454364e+00
8  6.820178962415187840e+00
9  7.761551232444479531e+00
10 8.675957032941749247e+00
11 9.556506054671928041e+00
12 1.041679752760606092e+01
13 1.124257868945134575e+01
14 1.204678297035635026e+01
15 1.285155385842977793e+01
16 1.363945376809673782e+01
17 1.442999510300424681e+01
18 1.521663219602045913e+01
19 1.599189823692343637e+01
20 1.676101874564537653e+01
21 1.752772897918726969e+01
22 1.828869938626520764e+01
23 1.904654053306329331e+01
24 1.980145734299838978e+01
25 2.055595680427493477e+01
26 2.130626699765168297e+01
27 2.205552903495953387e+01
};
\addlegendentry{SSC}

\addplot[
color=cyan,
thick,
]
table {
0  0
1  0
2  0
3  0
4  0
5  3.169925001442312151e+00
6  3.169925001442312151e+00
7  4.392317422778760694e+00
8  5.209453365628950117e+00
9  6.108524456778169132e+00
10 6.894817763307943714e+00
11 8.016808287686554735e+00
12 8.675957032941749247e+00
13 9.525520809095070263e+00
14 1.035204342579543280e+01
15 1.111569395219701128e+01
16 1.190951811429661511e+01
17 1.268759438957612495e+01
18 1.345288469748661697e+01
19 1.423144620918397685e+01
20 1.497356284719469599e+01
21 1.573843536569174439e+01
22 1.648181519529292771e+01
23 1.723945868307927398e+01
24 1.798214131144002437e+01
25 1.873002891306493822e+01
26 1.947577498446823796e+01
27 2.021757555754692959e+01
};
\addlegendentry{Fast-SSC}

\addplot[
color=blue,
dashed,
domain=10:27,
samples=10,
smooth,
]
{(1-1/3.63)*x+2.8};

\addplot[
color=cyan,
dashed,
thick,
domain=10:27,
samples=10,
smooth,
]
{(1-1/3.63)*x+1};

\end{axis}
\end{tikzpicture}

%% file: figures/resultsAWGNAll.tex
\begin{tikzpicture}

\begin{axis}[
scale=1,
xmin=0,
xmax=27,
ymin=0,
ymax=30,
grid=both,
ymajorgrids=true,
xmajorgrids=true,
grid style=dashed,
width=\textwidth, height=7.5cm,
xlabel={$n$},
ylabel={$\log_2 \mathcal{L}$},
ylabel shift=-7,
legend cell align={left},
legend pos=north west,
]

\node (comment1) at (axis cs:12,25){slope $=1$};
\draw[->] (comment1) -- (axis cs:25,26);

\node (comment) at (axis cs:18,5){slope $=0.75$};
\draw[->] (comment) -- (axis cs:25,20);
\draw[->] (comment) -- (axis cs:26,19);

\addplot[
color=black,
thick,
]
table {
0  0
1  1.584963
2  2.807355
3  3.906891
4  4.954196
5  5.977280
6  6.988685
7  7.994353
8  8.997179
9  9.998590
10 10.999295
11 11.999648
12 12.999824
13 13.999912
14 14.999956
15 15.999978
16 16.999989
17 17.999994
18 18.999997
19 19.999999
20 20.999999
21 22.000000
22 23.000000
23 24.000000
24 25.000000
25 26.000000
26 27.000000
27 28.000000
};

\addplot[
color=blue,
thick,
]
table {
0  0
1  0
2  0
3  0
4  0
5  3.459431618637297312e+00
6  4.087462841250339096e+00
7  5.285402218862248702e+00
8  6.507794640198696356e+00
9  7.451211111832328982e+00
10 8.388017285345135576e+00
11 9.269126679149417569e+00
12 1.013057056280542589e+01
13 1.098228060455828370e+01
14 1.182137539237253954e+01
15 1.265664831174735028e+01
16 1.347078595576430082e+01
17 1.427663333088670328e+01
18 1.506865233891319278e+01
19 1.585941345668450708e+01
20 1.663703519100529604e+01
21 1.740647907283767992e+01
22 1.817561289715512629e+01
23 1.894365089933583945e+01
24 1.970479356355991385e+01
25 2.046368899950802600e+01
26 2.122038009912986212e+01
27 2.197507091502399845e+01
};

\addplot[
color=cyan,
thick,
]
table {
0  0
1  0
2  0
3  0
4  0
5  0
6  3.169925001442312151e+00
7  3.459431618637297312e+00
8  5.044394119358453388e+00
9  5.832890014164741288e+00
10 6.714245517666122431e+00
11 7.607330313749610440e+00
12 8.519636252843213242e+00
13 9.359749560322329742e+00
14 1.013057056280542589e+01
15 1.096938652134228320e+01
16 1.176528609872368314e+01
17 1.252869828357418847e+01
18 1.330904996251514483e+01
19 1.408704840470706010e+01
20 1.485754368965403849e+01
21 1.562156542291651995e+01
22 1.638260714100043458e+01
23 1.713614002656481006e+01
24 1.789395915545862792e+01
25 1.864348631087148789e+01
26 1.939359365732201113e+01
27 2.014168183648478205e+01
};

\addplot[
color=blue,
dashed,
domain=10:27,
samples=10,
smooth,
]
{(1-1/3.63)*x+2.7};

\addplot[
color=cyan,
dashed,
thick,
domain=10:27,
samples=10,
smooth,
]
{(1-1/3.63)*x+1};

\end{axis}
\end{tikzpicture}

%% file: figures/resultsBSCAll.tex
\begin{tikzpicture}

\begin{axis}[
scale=1,
xmin=0,
xmax=27,
ymin=0,
ymax=30,
grid=both,
ymajorgrids=true,
xmajorgrids=true,
grid style=dashed,
width=\textwidth, height=7.5cm,
xlabel={$n$},
ylabel={$\log_2 \mathcal{L}$},
ylabel shift=-7,
legend cell align={left},
legend pos=north west,
]

\node (comment1) at (axis cs:12,25){slope $=1$};
\draw[->] (comment1) -- (axis cs:25,26);

\node (comment) at (axis cs:18,5){slope $=0.76$};
\draw[->] (comment) -- (axis cs:25,20);
\draw[->] (comment) -- (axis cs:26,19);

\addplot[
color=black,
thick,
]
table {
0  0
1  1.584963
2  2.807355
3  3.906891
4  4.954196
5  5.977280
6  6.988685
7  7.994353
8  8.997179
9  9.998590
10 10.999295
11 11.999648
12 12.999824
13 13.999912
14 14.999956
15 15.999978
16 16.999989
17 17.999994
18 18.999997
19 19.999999
20 20.999999
21 22.000000
22 23.000000
23 24.000000
24 25.000000
25 26.000000
26 27.000000
27 28.000000
};

\addplot[
color=blue,
thick,
]
table {
0  0
1  0
2  0
3  0
4  0
5  3.459431618637297312e+00
6  4.087462841250339096e+00
7  5.285402218862248702e+00
8  6.409390936137701722e+00
9  7.483815777264256397e+00
10 8.370687406807217457e+00
11 9.197216693110052077e+00
12 1.006474276475025675e+01
13 1.094909715572920739e+01
14 1.181498293642679087e+01
15 1.262182294739420207e+01
16 1.344047989881121907e+01
17 1.423369456554691226e+01
18 1.503217490229382491e+01
19 1.581545835083124629e+01
20 1.660074050438925752e+01
21 1.737353640716244740e+01
22 1.814266052948796215e+01
23 1.891275779561843606e+01
24 1.967743921421040554e+01
25 2.043903848732817252e+01
26 2.119554943325594465e+01
27 2.195090576136351146e+01
};

\addplot[
color=cyan,
thick,
]
table {
0  0
1  0
2  0
3  0
4  0
5  0
6  3.169925001442312151e+00
7  3.169925001442312151e+00
8  4.954196310386874913e+00
9  5.882643049361841570e+00
10 6.539158811108030989e+00
11 7.451211111832328982e+00
12 8.487840033823051300e+00
13 9.296916206879288325e+00
14 1.011764310138909195e+01
15 1.092999806260902673e+01
16 1.169740200750270809e+01
17 1.249610477934012565e+01
18 1.327160892415174409e+01
19 1.405095396520803241e+01
20 1.481683366794566936e+01
21 1.559090862653719789e+01
22 1.635226762216943897e+01
23 1.711505381338561094e+01
24 1.786683820831663283e+01
25 1.861866151632295896e+01
26 1.936772827605014413e+01
27 2.011689852254886546e+01
};

\addplot[
color=blue,
dashed,
domain=10:27,
samples=10,
smooth,
]
{(1-1/3.63)*x+2.7};

\addplot[
color=cyan,
dashed,
thick,
domain=10:27,
samples=10,
smooth,
]
{(1-1/3.63)*x+1};

\end{axis}
\end{tikzpicture}